\documentclass[envcountsame,envcountsect]{llncs}

\usepackage{amssymb}
\usepackage{amsmath}
\usepackage{tikz}
\usepackage{hyperref}
\usetikzlibrary{shapes,snakes,arrows,shadows,mindmap,trees}

\spnewtheorem{obs}[theorem]{Observation}{\bfseries}{\itshape}
\spnewtheorem{res}[theorem]{Restriction}{\bfseries}{\itshape}
\spnewtheorem*{thm}{Theorem}{\bfseries}{\itshape}

\newcommand{\T}{\mathcal{T}}
\newcommand{\A}{\mathcal{A}}
\newcommand{\B}{\mathcal{B}}
\newcommand{\C}{\mathcal{C}}

\newcommand{\bl}[1]{\mathbf{#1}}

\newcommand{\U}{\mathcal{U}}

\newcommand{\Su}{\mathcal{S}}
\newcommand{\X}{\mathcal{X}}

\newcommand{\su}{\mathbb{S}}

\begin{document}


\title{Lower Bounds for Restricted Schemes in the
	Two-Adaptive Bitprobe Model}

\author{Sreshth Aggarwal \and Deepanjan Kesh \and Divyam Singal}

\institute{Indian Institute of Technology Guwahati, Guwahati, Assam, India}

\maketitle


\begin{abstract}

In the adaptive bitprobe model answering membership queries in two
bitprobes, we consider the class of restricted schemes as introduced
by Kesh and Sharma~\cite{kesh:dam}. In that paper, the authors
showed that such restricted schemes storing subsets of size 2 require
$\Omega(m^\frac{2}{3})$ space. In this paper, we generalise the result
to arbitrary subsets of size $n$, and prove that the space required
for such restricted schemes will be
$\Omega(\left(\frac{m}{n}\right)^{1 - \frac{1}{\lfloor n / 4 \rfloor +
2}})$.

\end{abstract}

				\section{Introduction}

In the bitprobe model, we store subsets $\Su$ of size $n$ from an
universe $\U$ of size $m$ in a data structure taking $s$ amount of
space, and answer membership queries by reading $t$ bits of the data
structure. The conventional notation to denote such schemes is as an
$(n, m, s, t)$-scheme. Each such scheme has two components -- the {\em
storage scheme} sets the bits of the data structure according to the
given subset $\Su$, and the {\em query scheme} probes at most $t$ bits
of the data structure to answer membership queries. Schemes are
categorised as {\em adaptive} if the location of each bitprobe in
their query scheme depends on the answers obtained in the previous
bitprobes. If the location of the bitprobes in the query scheme is
independent of the answers obtained in the earlier bitprobes, the
corresponding scheme is called {\em non-adaptive}. For further reading
about bitprobe and other related models and their associated results,
Nicholson {\it et al.}~\cite{survey} has quite a detailed survey of
the area.

In this paper, we restrict ourselves to those bitprobe schemes that
answer membership queries using two adaptive bitprobes, i.e.\ $t=2$.
The data structure of such schemes can be thought of as having 3
tables, namely $\A, \B,$ and $\C$. The first bitprobe is made in table
$\A$, and if the bit probed in table $\A$ has been set to 0 the next
bitprobe is made in table $\B$. On the other hand, the second bitprobe
is made in table $\C$ if the bit queried in table $\A$ has been set to
1. The answer to the membership query is ``Yes'' if the second
bitprobe returns 1, ``No'' otherwise.

The best known scheme for storing subsets of size two and answering
membership queries using two adaptive bitprobes is due to
Radhakrishnan {\it et al.}~\cite{radha:raman} which takes
$O(m^\frac{2}{3})$ amount of space; the best known lower bound for the
problem is $\Omega(m^\frac{4}{7})$~\cite{radha:saswata}. Though the
problem is yet to be settled for subsets of size two, it has recently
been shown for subsets of size three that the space required is
$\Theta(m^\frac{2}{3})$~\cite{mirza:walcom,kesh:fsttcs}. Garg and
Radhakrishnan~\cite{mohit:few} proved that for arbitrary sized
subsets, the space bounds for two adaptive bitprobe schemes are
$\Omega(m^{1-\frac{1}{\lfloor n / 4 \rfloor}})$ and $O(m^{1 -
\frac{1}{4n+1}})$, where $n \leq c \cdot \log m$.

				\section{Restricted Schemes}

Kesh and Sharma~\cite{kesh:dam} proved that $\Omega(m^\frac{2}{3})$
is indeed the lower bound for two adaptive bitprobe schemes storing
subsets of size two, albeit for a restricted class of schemes. We now
introduce the restriction that characterises this class of schemes.

In the literature, elements of the universe $\U$ that query, or
equivalently map to, the same bit in table $\A$ are said to form a
{\em block}. We label the elements of a block uniquely as 1, 2, 3,
$\dots$, which we will refer to as the {\em index} of the element
within a block. The element with index $i$ of a block $\bl{a}$ will be
denoted as $\bl{a}_i$. Elements of $\U$ that query, or map to, the
same bit in tables $\B$ or $\C$ form a {\em set}. This departure in
labels is made to distinguish the collections of elements in tables
$\B$ and $\C$ from those of table $\A$, which will prove useful
henceforth. The set to which the element $\bl{a}_i$ belongs to in
table $\B$ will be denoted as $S_\B(\bl{a}_i)$; similarly for sets of
table $\C$.

We impose the following restriction on the schemes designed to store
subsets $\Su$ and answer membersip queries using two bitprobes.

\begin{res}

	If two elements belong to the same set either in table $\B$ or
	in table $\C$, then their indices are the same.

	\label{res:index}

\end{res}

To take an example, in the schemes that we consider if $\bl{a}_i \in
S_\C(\bl{b}_j)$, then it must be the case that $i = j$. We further
simplify our premise by imposing the following restrictions on the
schemes we are addressing. They are being made for the sake of
simplicity and do not affect the final result.

\begin{res}

	Our class of schemes satisfy the following constraints.

	\begin{enumerate}

		\item The three tables $\A, B,$ and $\C$ do not share
			any bit.

		\item All the three tables are of the same size.

		\item All the blocks in table $\A$ are of equal size.
			Let that size be $b$.

		\item There are no singleton sets in tables $\B$ and
			$\C$.

		\item All of the sets in the tables $\B$ and $\C$ are
			{\em clean}~\cite{kesh:fsttcs}, i.e.\ no two
			elements of a block belong to the same set.

	\end{enumerate}

	\label{res:simple}

\end{res}

As discussed in Section 1.5 of~\cite{kesh:dam}, the motivation for these kind of restrictions is from the schemes presented in such works as by
Radhakrishnan {\it et al.}~\cite{radha:raman}, Lewenstein {\it et
al.}~\cite{munro}, and Radhakrishnan {\it et
al.}~\cite{radha:saswata}. The final restriction is motivated from
Section~3 of Kesh~\cite{kesh:fsttcs}, where it is shown that any
scheme can be converted to a scheme with only clean sets with no
asymptotic increase in the size of the data structure.

The main result of the paper (Theorem~\ref{thm:lower}) is as follows.

\begin{thm}

	Two adaptive bitprobe schemes for storing subsets of size at most $n$ and satisfying Restriction~\ref{res:index} require $\Omega(\left(\frac{m}{n}\right)^{1 - \frac{1}{\lfloor n / 4 \rfloor + 2}})$ space.

\end{thm}

In this restricted setting, as one might expect, the lower bound
of Theorem~\ref{thm:lower} improves upon the bound proposed by Garg
and Jaikumar~\cite{mohit:few} for all schemes, which is
$\Omega(m^{1-\frac{1}{\lfloor n / 4 \rfloor}})$ for $n \leq c \cdot \sqrt{\frac{\log m}{\log n}}$, and the comparison can be found in Section~\ref{sec:lower}. To generalise the
proof presented, Lemma~\ref{lem:index}, which shows that indices
increase as the $i$ in $i$-Universe (Definition~\ref{def:univ}) increases, and
Lemma~\ref{lem:bad}, which works because the subsets $\Su$ and $\X$ (defined in Section~\ref{sec:premise})
are disjoint, need to be proven for the generalised setting. Other
lemmas, including those of Section~\ref{sec:lower}, lend itself to
generalisation without much effort.

				\section{Premise}

\label{sec:premise}

As mentioned earlier, the subset of the universe $\U$ that we want to
store in our data structure will be referred to as $\Su$. In the
subsequent discussion, it will be necessary to build certain subsets
that we would like to store in the data structure of the restricted
schemes and, consequently, arrive at certain contradictions -- such
subsets will be denoted at various places as $\Su, \Su', \Su'', \Su_1,
\Su_2,$ etc. As we build the subsets to store, it will also be
required to keep track of certain elements that cannot be part of
$\Su$ -- such subsets will be denoted as $\X, \X', \X'', \X_1, \X_2,$
etc.

As Restriction~\ref{res:index} forces sets in tables $\B$ and $\C$ to
contain elements of only a certain index, it will prove helpful to
refer to the various structures in the two tables by their indices. To
start with, $\U_i$ will denote those elements of $\U$ which have index
$i$. $\B_i$ will refer to the collection of all sets of table $\B$
comprised of elements of $\U_i$; similarly $\C_i$. Sometimes, we will
also use $\T$ and $\T'$ to refer to either of the tables $\B$ or $\C$.
Hence, if we have two distinct tables $\T$ and $\T'$, then one of them
will be $\B$ and the other $\C$, which is not important. On
the other hand, table $\A$ will always be referred to by its name.

In the literature the size of the data structure has always been
denoted by $s$. As the sizes of the three tables are equal
(Restriction~\ref{res:simple}), we would instead use $s$ to denote the
size of any particular table; this would alleviate the need of using
the fraction $\frac{s}{3}$ whenever we refer to the table sizes. So,
the schemes will henceforth be referred to as an $(n, m, 3 \cdot s,
2)$-schemes.

In the following two results, we present some self-evident and one
essential property of the notations as defined above. They will be
referenced to later, as needed.

\begin{obs}

	The size $s$ of a table and the elements of index $i$ are
	related as follows.

	\begin{enumerate}

		\item $|\A| = |\B| = |\C| = |\U_i| = s.$

		\item $\U_i = \bigcup\limits_{S \in \B_i} S$.

	\end{enumerate}

	\label{obs:prop}

\end{obs}

\begin{proof}

	The first of the two observations follows from the fact that
	each block of table $\A$ has exactly one element of any
	particular index. The second observation follows from the
	definition of $\B_i$. \qed

\end{proof}

\begin{lemma}

	The correctness of a scheme remains unaffected under a
	permutation of the indices.

	\label{lem:label}

\end{lemma}

\begin{proof}

	Consider an $(n, m, 3s, t)$-scheme that satisfies
	Restrictions~\ref{res:index} and \ref{res:simple}. Suppose
	$\pi$ is some permutation on the indices of the blocks of
	table $\A$. We observe that a permutation of the indices do
	not affect the membership of a block -- two elements which
	belonged to block $\bl{a}$ before, still belongs to $\bl{a}$
	but with their indices changed according to $\pi$. The same is
	true for any set of table $\B$ or $\C$ -- elements of a set
	all had the same index, say $i$, to start with, and they will
	now have the index $\pi(i)$. So, the data structure of a
	scheme remains unaffected under the permutation, only the
	labels of the sets have changed. Thus, if a scheme was correct
	to begin with, it will remain so after a permutation of the
	indices. \qed

\end{proof}

We end the section with a final notational convenience. We would, in
the discussion to follow, require to perform some arithmetic on index
$i$, like $i+1$ or $2i+2$. As the range of indices lie between 1 and
$b$, inclusive, all such expressions should be considered $\pmod b +
1$. This would help us to keep the expressions simple and avoid
repetition.

				\section{Nodes and Paths}

In this section we define {\em nodes}, {\em edges}, and {\em paths},
structures that are defined on top of the elements belonging to a set.

\begin{definition}

	A {\em node} of table $\T$, denoted as $(\bl{e}_k,
	\bl{f}_k)_\T$, is an ordered pair of distinct elements
	$\bl{e}_k$ and $\bl{f}_k$ such that they belong to the same
	set in $\T$.

	\label{def:node}

\end{definition}

Each of the components of a node are called its {\em terms}, the first
being referred to as the {\em antecedent} and the second as the {\em
consequent}.

We say that a block $\bl{a}$ is {\em stored} in table $\B$ if the bit
corresponding to the block $\bl{a}$ in table $\A$ is set to 0. Then
any query for any element of block $\bl{a}$ will be made in table $\B$, and
the sets corresponding to those elements in table $\B$ should be set
to 1 or 0 according as the elements are in $\Su$ or not. We can, hence,
say that the elements of block $\bl{a}$ are being stored in table
$\B$. Storing a block or an element in table in $\C$ can similarly be
defined as when the the bit in table $\A$ corresponding to the block
$\bl{a}$ is set to 1.

\begin{obs}

	Suppose a node be such that one of its terms is in the subset
	$\Su$ and the other in $\X$. Then, if the antecedent of the
	node is stored in its own table, the consequent of the node
	cannot be stored in its table.

	\label{obs:node}

\end{obs}

\begin{proof}

	Consider the node $(\bl{e}_k, \bl{f}_k)_\T$. If we store the
	antecedent in its table, namely $\T$, then there is way to
	ensure that the consequent cannot be stored in its table. To
	that end, we put $\bl{e}_k$ in $\Su$ and $\bl{f}_k$ in $\X$.
	Then as we are storing $\bl{e}_k$ in table $\T$, the set
	corresponding to $\bl{e}_k$ in the table must be set to 1. The
	element $\bl{f}_k$ belongs to the same set in $\T$ yet it is
	not part of $\Su$. So, $\bl{f}_k$, and consequently its block
	$\bl{f}$, cannot be stored in table $\T$, because if we do the
	query for element $\bl{f}_k$ will incorrectly return ``Yes''.

	An equally good choice to force the antecedent and the
	consequent to separate tables is to have $\bl{e}_k \in \X$ and
	$\bl{f}_k \in \Su$. \qed

\end{proof}

\begin{figure}[t]

    \centering
    
    \begin{tikzpicture}[thick,scale=0.8,every node/.style={transform shape}]
        
		\node at (0.5,7) {Table $\mathcal{B}$};
		
		\draw (-2,5.5) ellipse (.5cm and 1cm);
		\node at (-2,5.75) {$\bl{a}_1$};
		\node at (-2,5.15) {$\bl{b}_1$};
		\draw [decorate, decoration={brace,amplitude=10pt},xshift=-4pt,yshift=0pt] (-2.5,4.5) -- (-2.5,6.5) node[black,midway,xshift=-0.75cm,text width=.5cm,align=center] {Set $V$};
		
		\draw (0.5,5.5) ellipse (.5cm and 1cm);
		\node at (0.5,5.75) {$\bl{c}_3$};
		\node at (0.5,5.15) {$\bl{f}_3$};
		\draw [decorate, decoration={brace,amplitude=10pt},xshift=-4pt,yshift=0pt] (0,4.5) -- (0,6.5) node[black,midway,xshift=-0.75cm,text width=.5cm,align=center] {Set $W$};
		
		\draw (3,5.5) ellipse (.5cm and 1cm);
		\node at (3,6) {$\bl{d}_3$};
		\node at (3,5.5) {$\bl{e}_3$};
		\node at (3,5) {$\bl{g}_3$};
		\draw [decorate, decoration={brace,amplitude=10pt},xshift=-4pt,yshift=0pt] (2.5,4.5) -- (2.5,6.5) node[black,midway,xshift=-0.75cm,text width=.5cm,align=center] {Set $X$};

		\draw (3.8,4.3) -- (3.8,7.3);

		\node at (6.5,7) {Table $\mathcal{C}$};
		
		\draw (5.7,5.5) ellipse (.5cm and 1cm);
		\node at (5.7,6.2) {$\bl{b}_2$};
		\node at (5.7,5.733) {$\bl{c}_2$};
		\node at (5.7,5.266) {$\bl{d}_2$};
		\node at (5.7,4.8) {$\bl{e}_2$};
		\draw [decorate, decoration={brace,amplitude=10pt},xshift=-4pt,yshift=0pt] (5.2,4.5) -- (5.2,6.5) node[black,midway,xshift=-0.75cm,text width=.5cm,align=center] {Set $Y$};
		
		\draw (8.2,5.5) ellipse (.5cm and 1cm);
		\node at (8.2,5.75) {$\bl{g}_4$};
		\node at (8.2,5.15) {$\bl{h}_4$};
		\draw [decorate, decoration={brace,amplitude=10pt},xshift=-4pt,yshift=0pt] (7.7,4.5) -- (7.7,6.5) node[black,midway,xshift=-0.75cm,text width=.5cm,align=center] {Set $Z$};
		
		\node[ellipse,draw=black,minimum height=1.2cm,minimum width=0.6cm] (A) at (-2.5,0) {};
		\node at (-2.5,0.3) {$\bl{a}_1$};
		\node at (-2.5,-0.3) {$\bl{b}_1$};
		\node at (-2.2,-0.55) {$\mathcal{B}$};
		
		\node[ellipse,draw=black,minimum height=1.2cm,minimum width=0.6cm] (B) at (1,2.5) {};
		\node at (1,2.8) {$\bl{b}_2$};
		\node at (1,2.2) {$\bl{c}_2$};
		\node at (1.3,1.95) {$\mathcal{C}$};
		
		\node[ellipse,draw=black,minimum height=1.2cm,minimum width=0.6cm] (C) at (1,0.5) {};
		\node at (1,0.8) {$\bl{b}_2$};
		\node at (1,0.2) {$\bl{d}_2$};
		\node at (1.3,-0.05) {$\mathcal{C}$};
		
		\node[ellipse,draw=black,minimum height=1.2cm,minimum width=0.6cm] (D) at (1,-2.25) {};
		\node at (1,-1.95) {$\bl{b}_2$};
		\node at (1,-2.55) {$\bl{e}_2$};
		\node at (1.3,-2.8) {$\mathcal{C}$};
		
		\node[ellipse,draw=black,minimum height=1.2cm,minimum width=0.6cm] (E) at (4.5,3) {};
		\node at (4.5,3.3) {$\bl{c}_3$};
		\node at (4.5,2.7) {$\bl{f}_3$};
		\node at (4.8,2.45) {$\mathcal{B}$};
		
		\node[ellipse,draw=black,minimum height=1.2cm,minimum width=0.6cm] (F) at (4.5,1.5) {};
		\node at (4.5,1.8) {$\bl{d}_3$};
		\node at (4.5,1.2) {$\bl{e}_3$};
		\node at (4.8,0.95) {$\mathcal{B}$};
		
		\node[ellipse,draw=black,minimum height=1.2cm,minimum width=0.6cm] (G) at (4.5,0) {};
		\node at (4.5,0.3) {$\bl{d}_3$};
		\node at (4.5,-0.3) {$\bl{g}_3$};
		\node at (4.8,-0.55) {$\mathcal{B}$};
		
		\node[ellipse,draw=black,minimum height=1.2cm,minimum width=0.6cm] (H) at (4.5,-1.5) {};
		\node at (4.5,-1.2) {$\bl{e}_3$};
		\node at (4.5,-1.8) {$\bl{g}_3$};
		\node at (4.8,-2.05) {$\mathcal{B}$};
		
		\node[ellipse,draw=black,minimum height=1.2cm,minimum width=0.6cm] (I) at (4.5,-3) {};
		\node at (4.5,-2.7) {$\bl{e}_3$};
		\node at (4.5,-3.3) {$\bl{d}_3$};
		\node at (4.8,-3.55) {$\mathcal{B}$};
		
		\node[ellipse,draw=black,minimum height=1.2cm,minimum width=0.6cm] (J) at (8,-0.75) {};
        \node at (8,-0.45) {$\bl{g}_4$};
		\node at (8,-1.05) {$\bl{h}_4$};
		\node at (8.3,-1.3) {$\mathcal{C}$};
		
		\draw[line width=0.5pt, ->] (A) -- (B) node[midway, above] {$\bl{b}$};
		\draw[line width=0.5pt, ->] (A) -- (C) node[midway, above] {$\bl{b}$};
		\draw[line width=0.5pt, ->] (A) -- (D) node[midway, above] {$\bl{b}$};
		\draw[line width=0.5pt, ->] (B) -- (E) node[midway, above] {$\bl{c}$};
		\draw[line width=0.5pt, ->] (C) -- (F) node[midway, above] {$\bl{d}$};
		\draw[line width=0.5pt, ->] (C) -- (G) node[midway, above] {$\bl{d}$};
		\draw[line width=0.5pt, ->] (D) -- (H) node[midway, above] {$\bl{e}$};
		\draw[line width=0.5pt, ->] (D) -- (I) node[midway, above] {$\bl{e}$};
		\draw[line width=0.5pt, ->] (G) -- (J) node[midway, above] {$\bl{g}$};
		\draw[line width=0.5pt, ->] (H) -- (J) node[midway, above] {$\bl{g}$};

	\end{tikzpicture}
	
    \caption{Nodes, edges, and paths for the example arrangement of elements in sets.}
    \label{fig:nodes-paths}
\end{figure}
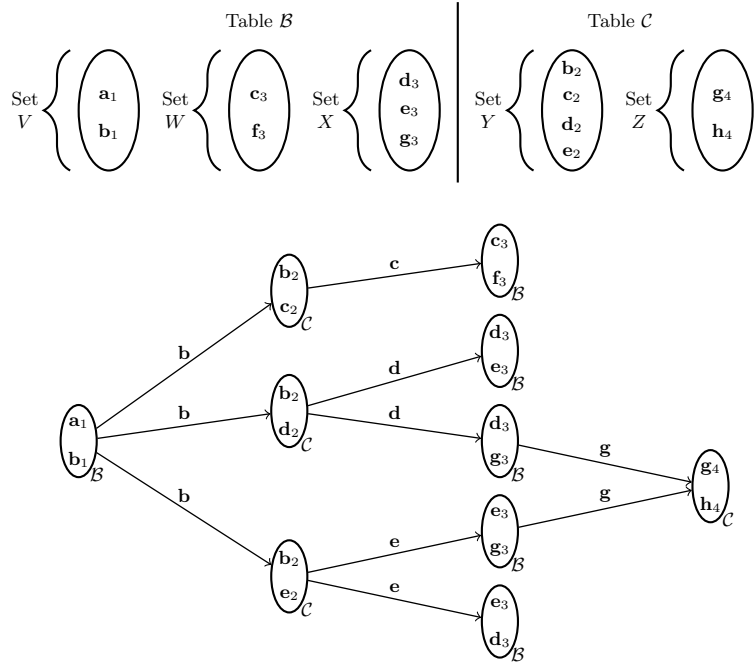

\begin{definition}

	There is said to be an {\em edge} from the node $(\bl{e}_k,
	\bl{f}_k)_{\T_1}$ to the node $(\bl{g}_l, \bl{h}_l)_{\T_2}$ if
	the following holds.

	\begin{enumerate}

		\item The nodes belong to distinct tables, i.e.\ $\T_1
			\neq \T_2$.

		\item $l = k+1$.

		\item The consequent of the first node and the
			antecedent of the second node belong to the
			same block, i.e.\ $\bl{f} = \bl{g}$.

	\end{enumerate}

	\label{def:edge}

\end{definition}

The second node above can be rewritten as $(\bl{f}_{k+1},
\bl{h}_{k+1})_{\T_2}$. The nodes with the edge between them are
connected via the common block $\bl{f}$, and hence will be will be
shown as
\[
	(\bl{e}_k, \bl{f}_k)_{\T_1} \
		\overset{\bl{f}}{\longrightarrow} \
		(\bl{f}_{k+1}, \bl{h}_{k+1})_{\T_2}.
\]

\begin{definition}

	A sequence of nodes is said to be a {\em path} if between
	every pair of adjacent nodes there is an edge from the former
	to the latter. The {\em length} of a path is the number of
	edges it contains.

	\label{def:path}

\end{definition}

A path will be denoted as
\[
	(\bl{e}_k, \bl{f}_k)_{\T_1} \
		\overset{\bl{f}}{\longrightarrow} \
		(\bl{f}_{k+1}, \bl{g}_{k+1})_{\T_2} \
		\overset{\bl{g}}{\longrightarrow} \
		(\bl{g}_{k+2}, \bl{h}_{k+2})_{\T_1} \
		\overset{\bl{h}}{\longrightarrow} \
		\dots
\]

To take an example, Figure~\ref{fig:nodes-paths} shows the relevant nodes, edges, and paths in the following scenario. In table $\mathcal{B}$, the elements $\bl{a}_1$ and $\bl{b}_1$ belong to the set $V$. Similarly, the elements $\bl{c}_3$, $\bl{f}_3$ belong to the set $W$, and $\bl{d}_3$, $\bl{e}_3$, $\bl{g}_3$ belong to the set $X$. In table $\mathcal{C}$, $\bl{b}_2$, $\bl{c}_2$, $\bl{d}_2$, $\bl{e}_2$ belong to the set $Y$, and $\bl{g}_4$, $\bl{h}_4$ belong to the set $Z$. One such node in the graph is $(\bl{a}_1, \bl{b}_1)_\mathcal{B}$. The node $(\bl{a}_1, \bl{b}_1)_\mathcal{B}$ is connected to $(\bl{b}_2, \bl{c}_2)_\mathcal{C}$ via the common block $\bl{b}$, is an edge in the graph. As is evident from the figure, there exists two paths from the node $(\bl{a}_1, \bl{b}_1)_\mathcal{B}$ to the node $(\bl{g}_4, \bl{h}_4)_\mathcal{C}$, each of length three.

For our discussion, we will only consider paths of length at most
$\lfloor \frac{b}{2} \rfloor - 1$; we will see in
Section~\ref{sec:modified} as to the reason why.

\begin{obs}

	Any element occurs at most once in a path.

	\label{obs:path}

\end{obs}

\begin{proof}

	It follows from the definition of a path which dictates that
	indices increase $\pmod b + 1$ from the first node onwards,
	and from our upper bound on the length of a path. On the other
	hand, it should be noted that a block may occur multiple times
	along a path. \qed

\end{proof}

\begin{lemma}

	Suppose for every node in a path one of the terms of the node
	is in $\Su$ and the other is in $\X$. Then, if the antecedent of
	the first node is stored in its own table, antecedents of all
	the nodes will have to be stored in their respective tables
	and the consequents of the nodes cannot be stored in their
	respective tables.

	\label{lem:path}

\end{lemma}

\begin{proof}

	The lemma is a direct consequence of
	Observation~\ref{obs:node}, and the proof can be found in the
	appendix as Lemma~\ref{app:lem:path}. \qed

\end{proof}

				\section{Universe of Elements}

In this section, we define the universe of an element of $\U$
recursively, and establish its relation with nodes and paths.

\begin{definition}

	The {\em $i$-Universe} of an element $\bl{e}_k$ w.r.t.\ table
	$\T$, denoted as $\U^i_\T(\bl{e}_k)$, is defined as follows.
	\[
		\U^i_\T(\bl{e}_k) =
			\begin{cases}
				\left\{ \ \bl{u}_{k+1} \ \mid \
					\bl{u}_k \in S_\T(\bl{e}_k) \setminus
					\{ \bl{e}_k \} \ \right\},
					& \textnormal{for $i = 1$;} \\[1em]
				\bigcup\limits_{\bl{u}_l \ \in
					\ \U^{i-1}_\T(\bl{e}_k)}
					\U^1_{\T'}(\bl{u}_l),
					& \textnormal{for $i>1$.}
			\end{cases}
	\]
	The table $\T'$ is defined as follows.
	\[ \begin{array}{cl}
		\T' = \T, & \textnormal{if $i$ is odd}; \\
		\T' \neq \T, & \textnormal{otherwise.}
	\end{array} \]

	\label{def:univ}

\end{definition}

Similar to the upper bound on paths, we will consider $i$-universes
for $1 \leq i \leq \lfloor \frac{b}{2} \rfloor - 1$, and, as stated
before, we will see in Section~\ref{sec:modified} as to the reason why.

The $i$-Universe of an element is the union of the 1-Universes of all
the elements in its $(i-1)$-Universe. So, as $i$ increases so does the
size of the universe. We will show that the elements of the
$i$-Universe must necessarily belong to distinct sets in a table, so
the larger the $i$ the larger has to be the size of the data structure
to accomodate the $i$-Universe. We start with a few properties of the
elements belonging to the $i$-Universe of an element.

\begin{obs}
	$\left| \U^1_\T(\bl{e}_k) \right| \ = \ \left| S_\T(\bl{e}_k) \setminus \{ \bl{e}_k \} \right|$.
	\label{obs:1univ}
\end{obs}

\begin{lemma}

	If an element $\bl{x}_l$ belongs to the $i$-Universe of
	$\bl{e}_k$, then $l = k + i$.

	\label{lem:index}

\end{lemma}

\begin{proof}

	The statement of the lemma can be established by induction on
	$i$, and the proof is given in the appendix as
	Lemma~\ref{app:lem:index}. \qed

\end{proof}

\begin{lemma}

	If the element $\bl{x}_{k+i}$ belongs to the $i$-Universe of
	$\bl{e}_k$ w.r.t.\ table $\T$, then there is a path such that

	\begin{enumerate}

		\item The first node is in table $\T$ with its
			antecedent being $\bl{e}_k$.

		\item The last node is in table $\T'$ with its
			antecedent being $\bl{x}_{k+i}$. The table
			$\T'$ is defined as follows.
			\[ \begin{array}{rl}
				\T = \T', & \textnormal{if $i$ is even} \\
				\T \neq \T', & \textnormal{otherwise}.
			\end{array} \]

		\item The length of the path is $i$.

	\end{enumerate}

	\label{lem:univ}

\end{lemma}

It is important to observe that the nature of the table $\T'$ in the
lemma above is contrary to that in Definition~\ref{def:univ} in that
$\T'$ is the same as $\T$ when $i$ is even in the lemma above, whereas
they are equal when $i$ is odd in the definition of the $i$-Universe.

\begin{proof}

	We will prove the aforementioned statement by induction on
	$i$ in the appendix in Lemma~\ref{app:lem:univ}. \qed

\end{proof}

				\section{Bad Elements}

In this section, we show that large universes of elements give rise
to {\em bad elements}, which put constraints on how and what subsets
can be stored.

\begin{definition}

	An element $\bl{e}_k$ is said to be {\em $i$-bad} w.r.t.\
	table $\T$ if for any $j$ between 1 and $i$, inclusive, there
	exist distinct elements $\bl{u}_{k+j}$ and
	$\bl{v}_{k+j}$ in $\U^j_\T(\bl{e}_k)$ s.t.\
	\[
		\bl{v}_{k+j} \in S_{\T'}(\bl{u}_{k+j}).
	\]
	The table $\T'$ is defined as follows.
	\[ \begin{array}{rl}
		\T = \T', & \textnormal{if $j$ is even} \\
		\T \neq \T', & \textnormal{otherwise}.
	\end{array} \]

	Elements which are not $i$-bad are said to be {\em $i$-good}.

	\label{def:bad}

\end{definition}

The above definition suggests that if $l$ is some constant less than
or equal to $i$ and the element $\bl{e}_k$ is $l$-bad, then it is also
$i$-bad.

\begin{lemma}

	If an element $\bl{e}_k$ is $i$-bad w.r.t.\ table $\T$, then
	there exists a choice of the sets $\Su$ and $\X$, each of size
	at most $2i$, s.t.\ the block $\bl{e}$ cannot be stored in
	table $\T$.

	\label{lem:bad}

\end{lemma}

\begin{proof}

	This lemma has been proved in the appendix as Lemma~\ref{app:lem:bad}. \qed

\end{proof}

				\section{Modified Schemes}
				\label{sec:modified}

Consider any restricted adaptive $(n, m, 3s, 2)$-scheme, the last
component 2 denoting the number of bitprobes allowed. Let some element
$\bl{e}_1$ of its universe $\U$ be $i$-bad w.r.t.\ table $\B$.
Lemma~\ref{lem:bad} states that there exist sets $\Su_1$ and $\X_1$,
each of size at most $2i$, s.t.\ the block $\bl{e}$ cannot be stored
in table $\B$. Also, as an element becomes $i$-bad due to the elements
of its $i$-Universe, the indices of the elements in the either of the
sets $\Su_1$ and $\X_1$ lie between 1 and $i+1$. Consider the element
$\bl{e}_{i+2}$. If this element is $i$-bad w.r.t.\ table $\C$ there
will exist sets $\Su_2$ and $\X_2$, again of size at most $2i$ each,
s.t.\ the block $\bl{e}$ cannot be stored in table $\C$. The range of
the indices in the two sets in this case would be from $i+2$ to
$2i+2$.

We already know that the sets $\Su_1$ and $\X_1$ are disjoint, as are
the sets $\Su_2$ and $\X_2$. Furthermore, as the range of indices in
the two pairs of sets do not overlap, we can deduce that all the four
sets are disjoint. Let us then consider the sets
\[
	\Su = \Su_1 \cup \Su_2 \
		\textnormal{ and } \
		\X = \X_1 \cup \X_2,
\]
each of their sizes being at most $4i$. As discussed above, this pair
of sets imply that the block $\bl{e}$ cannot be stored in either of
the tables $\B$ or $\C$, which is absurd as the scheme is deemed to be
correct. So, we may conclude the following.

\begin{lemma}

	For any block of table $\A$, say $\bl{e}$, if the element
	$\bl{e}_1$ is $i$-bad w.r.t.\ table $\B$, then the element
	$\bl{e}_{i+2}$ cannot be $i$-bad w.r.t.\ table $\C$.

	\label{lem:good}

\end{lemma}

Let us partition the universe $\U$ based on good and bad elements
w.r.t.\ table $\B$. One part will be the union of all those blocks
whose index 1 elements are good. The other part will be union of the
remaining blocks.
\[
	\U' = \mathop{\bigcup_{\bl{a}_1
		\textnormal{ is $i$-good}}}_{\textnormal{w.r.t.\ }\B}
		\bl{a}; \ \
	\U'' = \mathop{\bigcup_{\bl{a}_1
		\textnormal{ is $i$-bad}}}_{\textnormal{w.r.t.\ }\B}
		\bl{a}.
\]
According to Lemma~\ref{lem:good}, we know that though the index $1$
elements of the blocks of $\U''$ are bad w.r.t.\ table $\B$, the index
$i+2$ elements must necessarily be good w.r.t.\ table $\C$.

We now split our data structure in the following way. For any set $X$
in either of table $\B$ or $\C$, we split it into two sets, one
containing the elements of $\U'$ and the other containing the elements of
$\U''$. More formally,
\[
	X = X' \cup X''; \ \ X' \subset \U', \ X'' \subset \U''.
\]
It is important to note that the indices of the elements in the two sets
$X'$ and $X''$ are the same as that of $X$. Consequently, the table
$\B$ has been split into two parts, namely $\B'$ containing the sets
with elements from $\U'$, and $\B''$ containing sets with elements
from $\U''$. Thus, the collection of all sets in table $\B$ containing
elements with index $k$, namely $\B_k,$ is now $\B'_k \cup \B''_k$.
The table $\C$ have similarly been split into two parts - $\C'$ and
$\C''$.

\begin{obs}

	The size of $\B_k$ has at most doubled due to the above
	modification.

	\label{obs:split}

\end{obs}

The table $\A$ is also split into two tables, namely $\A'$ and $\A''$,
containing elements of $\U'$ and $\U''$, respectively. As per our
definition of $\U'$ and $\U''$, either a block belongs entirely in
$\U'$ or entirely in $\U''$, and thus individual blocks are not split.

We now have two sets of data structures, one corresponding to the
elements of $\U'$ and the other corresponding to $\U''$ --
\[
	\left( \A', \B', \C' \right) \textnormal{ and }
		\left( \A'', \B'', \C'' \right).
\]
This also means that within the original scheme, we have two
independent schemes, one for the elements of $\U'$ and the other for
the elements of $\U''$. Any subset $\Su$ that is to be stored can now
be split into $\Su' \subset \U'$ and stored in the data structure
corresponding to $\U'$, and $\Su'' \subset \U''$ which can be stored
in the data structure corresponding to $\U''$. The storage and query
schemes remain as before for each of the parts of the data structure.
So, to store a subset, if any block was earlier set to 0, in the new
data structure it will still be set to 0. If any set $X$ was being set
to 1, now both $X'$ and $X''$ will be set to 1; and so on.

We further modify the new data structure as follows. For the part
$(\A'', \B'', \C'')$, we interchange the parts $\B''$ and $\C''$ in
the tables $\B$ and $\C$ so that $\B''$ will now be part of table $\C$
and $\C''$ will now be part of table $\B$. With this modification, for
any index $k$ the tables will be as follows.
\[
	\B_k = \B'_k \cup \C''_k; \ \ \C_k = \C'_k \cup \B''_k
\]
As the part pertaining to $\U'$, i.e.\ $(\A', \B', \C')$, is
unaffected, the query scheme and storage scheme for it remains
unchanged. For the part $(\A'', \B'', \C'')$, if a block was earlier
set to 0 and thus sent to table $\B$, it should now be set to 1 and
sent to table $\C$. Similarly, a block which was earlier set to 1 will
now have to be set to 0.

Let us consider the sizes of the tables. For lack of a better
notation, we will use $\T^{(0)}, \T^{(1)}, \T^{(2)}$ to refer to the
original table, the table after the first modification, and after the
second modification, respectively. Observation~\ref{obs:split} tells
us that
\[
	|\T_k^{(1)}| \leq 2 \cdot |\T_k^{(0)}|.
\]
After the second modification, we note that
\[
	|\B_k^{(2)}| + |\C_k^{(2)}| = |\B_k^{(1)}| + |\C_k^{(1)}|.
\]

We make the third and final modification to our scheme. Before the
second modification, all the elements of index 1 in $\B'$ were good
w.r.t.\ table $\B$. After the second modification, all the elements of index
$i+2$ in $\C''$, which were earlier good w.r.t.\ table $\C$, are now
good w.r.t.\ table $\B$ because $\C''$ is now part of table $\B$. In
Lemma~\ref{lem:label}, we have seen that the correctness of a scheme
remains unaffected under a permutation of its indices. We now apply
the following permutation over the indices of the data structure
corresponding to $\U''$ -- the labels $k$ and $k+i+1$ are
interchanged, where $1 \leq k \leq i+1$, whereas the rest of the
indices remain unchanged. With this further modification, all the
indices labelled from $i+2$ to $2i+2$ will now be labelled 1 to $i+1$
in that order, whereas the previously labelled indices 1 and $i+1$ will
now be labelled $i+2$ to $2i+2$, in that order.

With this final modification, we now have a scheme where all the elements
of index 1 are good w.r.t.\ table $\B$.

\begin{lemma}

	Any given restricted $(n, m, 3s, 2)$-scheme can be modified
	into a $(n, m, 6s, 2)$-scheme such that in the modified
	scheme all the elements of index 1 are $i$-good w.r.t.\ table
	$\B$.

	\label{lem:modified}

\end{lemma}

As the third and final modification does not affect indices larger
than $2i+2$, we can say that
\begin{equation*}
	|\B_k^{(3)}| + |\C_k^{(3)}|
		\ = \ |\B_k^{(2)}| + |\C_k^{(2)}|
		\ \leq \ 2 \cdot (|\B_k^{(0)}| + |\C_k^{(0)}|),
\end{equation*}
for $k > 2i+2$. As for the indices 1 to $2i+2$, the sets have been
relabelled but not created, and as a result the total number of sets
remain unchanged, i.e.\
\begin{equation}
	\sum_{k=1}^{2i+2} \left( |\B_k^{(3)}| + |\C_k^{(3)}| \right)
		\ = \ \sum_{k=1}^{2i+2} \left( |\B_k^{(2)}|
			+ |\C_k^{(2)}| \right) \\
		\ \leq \ 2 \cdot \sum_{k=1}^{2i+2} \left( |\B_k^{(0)}|
			+ |\C_k^{(0)}| \right).
	\label{eqn:size}
\end{equation}

Finally, all of this can only be proven for subsets $\Su$ and $\X$
whose sizes are at least $4i$, and for the range of indices 1 to
$2i+2$. So, for the first condition we can set $n = 4i$. As for the range of indices, the
size of a block $b$ has to be larger than $2i+2$, which implies that
universes and path lengths are bounded by $\lfloor \frac{b}{2} \rfloor
- 1$.

				\section{Lower Bound}
				\label{sec:lower}

In this section, we will present our theorem on the space lower bound
on restricted schemes. We start by presenting an estimate of the total sizes of all the $t$-universes of good elements.

\begin{lemma}

	Suppose all the elements with index 1 are $t$-good w.r.t.\
	table $\B$. Then the sizes of their $t$-Universes satisfy
	the following inequality.
	\[
		\sum_{\bl{e} \ \in \ \A} \left| \U^t_\B(\bl{e}_1)
			\right|
		\ \ \geq \ \ c \cdot \frac{s^{t+1}}{\left( \sum_{i=1}^{t}
			(|\B_i| + |\C_i|) \right)^t},
	\]
	for some constant $c$.

	\label{lem:univ_size}

\end{lemma}

\begin{proof}
In this lemma, for the sake of convenience, we introduce two new
notations. If $\bl{h}$ is a block then $\bl{h}_i$ was meant to denote that element of $\bl{h}$ which has index $i$. We now abuse the notation and use $\bl{h}_{i,j}$ to denote the element with index $j$ in
the block $\bl{h}_i$. We also introduce $P_\T(\bl{e}_k)$ to denote the set
$S_\T(\bl{e}_k) \setminus \{ \bl{e}_k \}$. These notations will help
us keep the expressions to follow succint.

Let us assume that $t$ is odd. The sum of the sizes of the $t$-universes of all index 1 elements can be expressed as follows.
\begin{align*}
	\sum_{\bl{e} \ \in \ \A} \left| \U^t_\B(\bl{e}_1) \right|
		&= \sum_{\bl{e} \in \A} \left| \left(
			\bigcup_{\bl{h}_{t-1,t} \in
			\U^{t-1}_\B(\bl{e}_1)} \U^1_\B(\bl{h}_{t-1,t})
			\right) \right|
		= \sum_{\bl{e} \in \A} \left( \sum_{\bl{h}_{t-1,t} \in
			\U^{t-1}_\B(\bl{e}_1)} \left|
			\U^1_\B(\bl{h}_{t-1,t}) \right| \right) \\
		&= \sum_{\bl{e} \ \in \ \A} \ \ \sum_{\bl{h}_{t-1,t} \
			\in \ \U^{t-1}_\B(\bl{e}_1)} \left|
			P_\B(\bl{h}_{t-1,t}) \right|
\end{align*}
The above derivation follows from the definition of $t$-universe, the fact that all elements of index 1 are $t$-good, and Observation~\ref{obs:1univ} about the size of 1-universes. We have now arrived at a
summation indexed by the elements of $\U^{t-1}_\B(\bl{e}_1)$, and
applying Lemma~\ref{app:lem:identity}, we get --
\begin{multline*}
	\sum_{\bl{e} \ \in \ \A} \left| \U^t_\B(\bl{e}_1) \right|
		\ \ = \ \ \sum_{\bl{e} \ \in \ \A} \ \
			\sum_{\bl{h}_{1,1} \ \in \ P_\B(\bl{e}_1)} \ \
			\sum_{\bl{h}_{2,2} \ \in \ P_\C(\bl{h}_{1,2})}
			\ \ \sum_{\bl{h}_{3,3} \ \in \
			P_\B(\bl{h}_{2,3})} \\
		\ \ \dots
		\ \ \sum_{\bl{h}_{t-1,t-1} \ \in \
			P_\C(\bl{h}_{t-2,t-1})} \left|
			P_\B(\bl{h}_{t-1,t}) \right|.
\end{multline*}
The summation $\sum_{\bl{e} \in \A}$ can be equivalently expressed as $\sum_{S \in \B_1} \sum_{\bl{e}_1 \in S}$. As the elements $\bl{e}_1$ and $\bl{h}_{1,1}$ both belong to the set $S$ and are distinct from each other, we can now reorder the first three indices of the summation as $\sum_{S \in \B_1} \sum_{\bl{h}_{1,1} \in S} \sum_{\bl{e}_1 \in P_\B(\bl{h}_{1,1})}$.
By pushing the summation indexed by $\bl{e}_1$ inside, we can finally rewrite down the
summation as --
\begin{multline}
	\sum_{\bl{e} \ \in \ \A} \left| \U^t_\B(\bl{e}_1) \right|
		\ \ = \ \ \sum_{\bl{h}_1 \ \in \ \A} \ \
			\sum_{\bl{h}_{2,2} \ \in \ P_\C(\bl{h}_{1,2})}
			\ \ \sum_{\bl{h}_{3,3} \ \in \
			P_\B(\bl{h}_{2,3})} \ \ \dots \\
		\ \ \sum_{\bl{h}_{t-1,t-1} \ \in \
			P_\C(\bl{h}_{t-2,t-1})} \left|
			P_\B(\bl{h}_{1,1}) \right| \cdot \left|
			P_\B(\bl{h}_{t-1,t}) \right|.
	\label{app:eqn:sum}
\end{multline}

Each term of this summation is determined by a tuple such as
\[
	\left( \bl{h}_1 \in \A, \ \bl{h}_{2,2} \in P_\C(\bl{h}_{1,2}),
		\ \bl{h}_{3,3} \in P_\B(\bl{h}_{2,3}), \ \dots, \
		\bl{h}_{t-1,t-1} \in P_\C(\bl{h}_{t-2,t-1}) \right)
\]
where each block, except for the first, is dependent on the previous
blocks. On the other hand, if any of the sets in the tuple is fixed,
then the other blocks and the terms of the summation they index are
fixed by the set. With this insight, we are going to put a lower bound
on the sum of all $t$-universes.

Suppose $X_1$ be the smallest set that occurs in the summation above
(Equation~\ref{app:eqn:sum}), either as one of its terms or as one of its
indices. We first consider the case when $X_1$ occurs as the index under
the $i$\textsuperscript{th} summation. Let us also consider that $i$
is odd, which would imply that $X_1$ belongs to table $\B$. Thus the
terms of the summation in which $X_1$ participates is determined as
follows -- the indices under the $i$\textsuperscript{th} summation and
beyond is determined as
\[
	\left( \ \bl{h}_{i,i} \in X_1, \ \bl{h}_{i+1,i+1} \in
		P_\C(\bl{h}_{i,i+1}), \ \bl{h}_{i+2,i+2} \in
		P_\B(\bl{h}_{i+1,i+2}), \ \dots \ \right),
\]
and the indices prior to that is determined as
\[
	\left( \ \bl{h}_{i-1,i} \in P_\B(\bl{h}_{i,i}), \
		\bl{h}_{i-2,i-1} \in P_\C(\bl{h}_{i-1,i-1}), \
		\ \dots,
	\ \bl{h}_{2,3} \in P_\B(\bl{h}_{3,3}), \ \bl{h}_{1,2}
		\in P_\C(\bl{h}_{2,2}) \ \right).
\]
It is important to note that in the latter of the two tuples,
$P_\B(\bl{h}_{i,i})$ is the set $X_1 \setminus \{ \bl{h}_{i,i} \}$.

The sum of all the terms in which the set $X_1$ participates is as follows.
\begin{multline*}
	\sum_{\bl{h}_{i-1,i} \in P_\B(\bl{h}_{i,i})}
		\ \ \sum_{\bl{h}_{i-2,i-1} \in P_\C(\bl{h}_{i-1,i-1})}
		\ \ \sum_{\bl{h}_{i-3,i-2} \in P_\B(\bl{h}_{i-2,i-2})}
		\dots
		\ \ \sum_{\bl{h}_{2,3} \in P_\B(\bl{h}_{3,3})}
		\ \ \sum_{\bl{h}_{1,2} \in P_\C(\bl{h}_{2,2})} \\
		\left(
		\ \ \sum_{\bl{h}_{i,i} \in X_1}
		\ \ \sum_{\bl{h}_{i+1,i+1} \in P_\C(\bl{h}_{i,i+1})}
		\dots \right.
		\left. \sum_{\bl{h}_{t-1,t-1} \in P_\C(\bl{h}_{t-1,t-2})}
		\left( \left| P_\B(\bl{h}_{1,1}) \right| \cdot \left|
			P_\B(\bl{h}_{t-1,t}) \right| \right) \right)
\end{multline*}
As all of the sets involved have sizes $\geq |X_1|$, the above sum is at least $c_1 \cdot |X_1|^{t+1}$, for some constant $c_1$. From the remaining terms and index sets of the summation in Equation~\ref{app:eqn:sum}, we remove all the blocks that belong to set $X_1$. So, if the initial sum in Equation~\ref{app:eqn:sum} is denoted by $\su_0$, and the remaining sum after the above procedure is $\su_1$, we have
\[
	\su_0 \ \ \geq \ \ \su_1 \ + \ c_1 |X_1|^{t+1}.
\]
We next identify the smallest set, say $X_2$, in the summation $\su_1$ and repeat the above proceduce which ends up in an estimation of all the terms associated with $X_2$, the estimation being $\geq c_2 \cdot |X_2|^{t+1}$, and removing the terms and blocks associated with $X_2$ from the remainder. We repeat this until all the blocks have thus been removed, upon which we will have a family of sets labelled $X_i$s and they partition the blocks of table $\A$. The number of sets would, in the worst case, be the total number of sets in the tables $\B$ and $\C$ with index at most $t$. Consequently, we have
\[
	\sum_{\bl{e} \ \in \ \A} \left| \U^t_\B(\bl{e}_1) \right|
		\geq c \cdot \sum_{i} \left| X_i \right|^{t+1}
		\geq c \cdot \sum_i \left( \frac{ \sum_i \left| X_i \right| }{\sum_i 1} \right)^{t+1}
		\geq c \cdot \frac{s^{t+1}}{\left( \sum_{i=1}^{t} (|\B_i| + |\C_i|) \right)^t},
\]
where $c$ is some suitable constant. The final bound arises using the Cauchy-Schwarz inequality.

All the other scenarios including the case where $t$ is presumed to be even, can be similarly argued. \qed
\end{proof}

\begin{lemma}

	If all elements of $\U_1$ are $t$-good w.r.t.\ table $\B$, then
	\[
		\sum_{j=1}^{t+1} \left( |\B_j| + |\C_j| \right)
			\geq c \cdot s^\frac{t}{t+1},
	\]
	for some constant $c$.

	\label{lem:size}

\end{lemma}

\begin{proof}

	As before, we will establish the statement of the lemma assuming that $t$ is odd. The case when $t$ is even will follow similarly. According to the definition of bad elements (Definition~\ref{def:bad}), a necessary property for an element to be $t$-good w.r.t.\ table $\B$ is that the elements of its $t$-Universe belong to distinct sets in $\C_{t+1}$, $t$ being odd. Consequently, we have
	\[
		\sum_{\bl{e}_1 \ \in \ \U_1} |\U^t_\B(\bl{e}_1)|
			\ \leq \ \sum_{\bl{e}_1 \ \in \ \U_1} |\C_{t+1}|
			\ = \ s \cdot |\C_{t+1}|
			\ \ \textnormal{ (Observation~\ref{obs:prop}) }
	\]
	From Lemma~\ref{lem:univ_size}, the inequalities follows.
	\begin{align*}
		s \cdot |\C_{t+1}|
			\ \geq \ \sum_{\bl{e}_1 \in \U_1} |\U^t_\B(\bl{e}_1)|
			\ &\geq \ \ c \cdot \frac{s^{t+1}}{\left( \sum_{i=1}^{t}
			(|\B_i| + |\C_i|) \right)^t} \\
		\implies \left( \sum_{i=1}^{t+1} \left( |\B_i| + |\C_i| \right) \right)^{t+1}
			&\geq c \cdot s^t,
	\end{align*}
	and the lemma follows. \qed

\end{proof}

Lemma~\ref{lem:modified} states that given a restricted $(n, m, 3s, 2)$-scheme, it can be modified into a $(n, m, 6s, 2)$-scheme such that in the modified scheme all the elements of $\U_1$ are $i$-good for some constant $i$. Furthermore, in that case we require the subset size, $n$, should be at least $4i$. So, from Equation~\ref{eqn:size} and Lemma~\ref{lem:size}, we can deduce the following.
\begin{align}
	\sum_{k=1}^{2i+2} \left( |\B_k^{(0)}| + |\C_k^{(0)}| \right)
		\ &\geq \ \frac{1}{2} \sum_{k=1}^{2i+2} \left( |\B_k^{(2)}| + |\C_k^{(2)}| \right)
		\ = \ \frac{1}{2}\sum_{k=1}^{2i+2} \left( |\B_k^{(3)}| + |\C_k^{(3)}| \right)
		\nonumber \\
		&\geq \ c \cdot (2s)^\frac{i}{i+1},
	\label{eqn:table}
\end{align}
where $2s$ comes from the fact that the first modification splits the sets of tables $\B$ and $\C$ (Observation~\ref{obs:split}).

Let the indices in the original scheme be so chosen that the sum on the first $2i+2$ indices in Equation~\ref{eqn:table} is the minimum among all choices. We can then derive the following.
\begin{align*}
	\sum_{k=1}^{b} \left( |\B_k^{(0)}| + |\C_k^{(0)}| \right)
		\ &\geq \ \frac{b}{2i+2} \sum_{k=1}^{2i+2} \left( |\B_k^{(0)}| + |\C_k^{(0)}| \right) \\
	\implies 2 \cdot s \ &\geq \ \frac{1}{2i+2} \frac{m}{s} \cdot c \cdot (2s)^\frac{i}{i+1},
\end{align*}
which upon simplification gives us
\[
	s \ \geq \ c' \cdot \left(\frac{m}{n}\right)^{1 - \frac{1}{\lfloor n/4 \rfloor + 2}},
\]
for some suitable constant $c'$. Hence, the main result of the paper is as follows.
\begin{theorem}

	Two adaptive bitprobe schemes for storing subsets of size at most $n$ and satisfying Restriction~\ref{res:index} require $\Omega(\left(\frac{m}{n}\right)^{1 - \frac{1}{\lfloor n / 4 \rfloor + 2}})$ space.

	\label{thm:lower}

\end{theorem}

Comparing our result in this restricted setting with the bound proposed by Garg
and Jaikumar~\cite{mohit:few} for all schemes, we see that our result improves on~\cite{mohit:few} for $n \leq c \cdot \sqrt{\frac{\log m}{\log n}}$.
\\
Our lower bound is better if the following holds --
\[ c_1 \cdot \left(\frac{m}{n}\right)^{1 - \frac{1}{\lfloor n / 4 \rfloor + 2}} \geq c_2 \cdot m^{1-\frac{1}{\lfloor n / 4 \rfloor}} \]
Taking logarithm on both sides, we have
\[ \left( \frac{1}{\frac{n}{4}} - \frac{1}{(\frac{n}{4}+2)} \right) \log m \geq c' \cdot \left( 1 - \frac{1}{(\frac{n}{4}+2)} \right) \log n \]
for some constant $c'$. \\
Upon further simplification,
\[ \log m \geq c' \cdot n^2 \log n \]
\[ n \leq c \cdot \sqrt{\frac{\log m}{\log n}} \]
for some constant $c$, and thus our claim holds.

				\section{Conclusion}

In this paper, we addressed a class of schemes, as devised by Kesh and
Sharma~\cite{kesh:dam}, in the two adaptive bitprobe model and
provided a space lower bound on such schemes for subsets of arbitrary
sizes, thereby generalising the lower bound presented in that paper. As discussed earlier, one of the key lemmas that our lower bound proof hinges upon is Lemma~\ref{lem:bad}, which demonstrates the generation of bad elements, and establishing this lemma is crucial in generalising the proof to arbitrary schemes. We hope that this issue can be resolved and the structure of our proof could serve as a template to provide bounds stronger that those presented by Garg and Jaikumar~\cite{mohit:few}.

				\appendix

\section{Appendix}

\begin{lemma}

	Suppose for every node in a path one of the terms of the node
	is in $\Su$ and the other in $\X$. Then, if the antecedent of
	the first node is stored in its own table, antecedents of all
	the nodes will have to be stored in their respective tables
	and the consequents of the nodes cannot be stored in their
	respective tables.

	\label{app:lem:path}

\end{lemma}

\begin{proof}

	It is to be noted that it is actually possible to put one
	element of each node in the set $\Su$ and the other in $\X$
	such that the two sets remain disjoint. It is due to the
	simple fact, summarised in Observation~\ref{obs:path}, that no
	element occurs more than once in a path.

	Let the path under consideration be
	\[
		(\bl{e}_k, \bl{f}_k)_{\T_1} \
			\overset{\bl{f}}{\longrightarrow} \
			(\bl{f}_{k+1}, \bl{g}_{k+1})_{\T_2} \
			\overset{\bl{g}}{\longrightarrow} \
			(\bl{g}_{k+2}, \bl{h}_{k+2})_{\T_1} \
			\overset{\bl{h}}{\longrightarrow} \
			\dots
	\]
	Furthermore, we are saving the element $\bl{e}_k$ in its own
	table, namely $\T_1$.

	Observation~\ref{obs:node} tells us that as one of the terms
	of the first node is in $\Su$ and the other in $\X$, the
	consequent $\bl{f}_k$ cannot be stored in $\T_1$. Thus the
	element $\bl{f}_k$ and its block $\bl{f}$ will have to be
	stored in table $\T_2$. For the second node the antecedent
	$\bl{f}_{k+1}$ is being stored in its own table, and, as
	before, one of its terms is in $\Su$ and the other in $\X$.
	So, the consequent of the second node cannot be stored in its
	own table $\T_2$. This again implies that the element
	$\bl{g}_{k+1}$ and its block $\bl{g}$ will have to be stored
	in $\T_1$, and thus the antecedent of the third node is being
	stored in its own table.

	Continuing thus, we see that for every node in the path the
	antecedent is being stored in its own table and the consequent
	is not, which establishes the statement of the lemma. \qed

\end{proof}

\begin{lemma}

	If an element $\bl{x}_l$ belongs to the $i$-Universe of
	$\bl{e}_k$, then $l = k + i$.

	\label{app:lem:index}

\end{lemma}

\begin{proof}

	We will prove the statement by using induction on $i$.
	Definition~\ref{def:univ} tells us that any element in
	$\U^1_\T(\bl{e}_k)$ has index $k+1$, so the lemma is trivially
	true for the base case.

	Suppose that the statement is true for all $i$ between 1 and
	$t$, inclusive. Let us now consider an element $\bl{x}_l$ in
	the $(t+1)$-Universe of $\bl{e}_k$. Then, according to
	Definition~\ref{def:univ} $\bl{x}_l$ belongs to the 1-Universe
	of some element in $\U^t_\T(\bl{e}_k)$. By induction
	hypothesis, the index of that element in the $t$-Universe is
	$k+t$. Let that element in $\U^t_\T(\bl{e}_k)$ be
	$\bl{u}_{k+t}$. As $\bl{x}_l$ belongs to the 1-Universe of
	$\bl{u}_{k+t}$, we have shown earlier that $l$ must be equal
	to $k+t+1$. This completes the induction. \qed

\end{proof}

\begin{lemma}

	If the element $\bl{x}_{k+i}$ belongs to the $i$-Universe of
	$\bl{e}_k$ w.r.t.\ table $\T$, then there is a path such that

	\begin{enumerate}

		\item The first node is in table $\T$ with its
			antecedent being $\bl{e}_k$.

		\item The last node is in table $\T'$ with its
			antecedent being $\bl{x}_{k+i}$. The table
			$\T'$ is defined as follows.
			\[ \begin{array}{rl}
				\T = \T', & \textnormal{if $i$ is even} \\
				\T \neq \T', & \textnormal{otherwise}.
			\end{array} \]

		\item The length of the path is $i$.

	\end{enumerate}

	\label{app:lem:univ}

\end{lemma}

\begin{proof}

	We will prove the aforementioned statement by induction on
	$i$. Consider the 1-Universe of $\bl{e}_k$ and $\bl{x}_{k+1}
	\in \U^1_{\T}(\bl{e}_k)$. By the definition of 1-Universe
	(Definition~\ref{def:univ}), $\bl{x}_k$ belongs to
	$S_\T(\bl{e}_k) \setminus \{ \bl{e}_k \}$. Furthermore,
	consider an element $\bl{y}_{k+1}$ in the set
	$S_{\T'}(\bl{x}_{k+1}) \setminus \{ \bl{x}_{k+1} \}$, where
	$\T' \neq \T$. The path to consider in such a scenario is
	\[
		(\bl{e}_k, \bl{x}_k)_{\T} \
			\overset{\bl{x}}{\longrightarrow} \
			(\bl{x}_{k+1}, \bl{y}_{k+1})_{\T'}.
	\]
	which satisfies all the requirements of the lemma. Thus the
	base case of the induction is proven.

	Let us assume that the statement of the lemma holds for all
	values of $i$ between 1 and $t$, inclusive. Let us further
	suppose that $t$ is odd. Consider an element $\bl{x}_{k+t+1}$
	in the $(t+1)$-Universe of $\bl{e}_k$. The index of the
	element is set by Lemma~\ref{lem:index}. According to the
	definition of $(t+1)$-Universe (Definition~\ref{def:univ}),
	there must exist an element $\bl{u}_{k+t}$ in the $t$-Universe
	of $\bl{e}_k$ s.t.\ $\bl{x}_{k+t+1} \in
	\U^1_{\T'}(\bl{u}_{k+t}).$ The index of the 1-Universe is $\T'
	\neq \T$ because of the definition of $(t+1)$-Universe and the
	fact that $t+1$ has been assumed to be even. In this scenario,
	we have already established that there will exist an edge of
	the form
	\[
		(\bl{u}_{k+t}, \bl{x}_{k+t})_{\T'} \
			\overset{\bl{x}}{\longrightarrow} \
			(\bl{x}_{k+t+1}, \bl{z}_{k+t+1})_{\T},
	\]
	for some $\bl{z}_{k+t+1} \in S_{\T}(\bl{x}_{k+t+1}) \setminus
	\{ \bl{x}_{k+t+1} \}$.

	As $\bl{u}_{k+t} \in \U^t_\T(\bl{e}_k)$, by induction
	hypothesis there exists a path of length $t$ of the form
	\[
		(\bl{e}_k, \bl{f}_k)_{\T} \
			\overset{\bl{f}}{\longrightarrow} \
			\dots \
			\overset{\bl{u}}{\longrightarrow} \
			(\bl{u}_{k+t}, \bl{y}_{k+t})_{\T'}.
	\]
	By induction hypothesis, the path will end in $\T' \neq \T$ as
	$t$ is odd.

	Combining the two paths above by setting $\bl{y}_{k+t} =
	\bl{x}_{k+t}$, we get the desired path of length $t+1$.
	\[
		(\bl{e}_k, \bl{f}_k)_{\T} \
			\overset{\bl{f}}{\longrightarrow} \
			\dots \
			\overset{\bl{u}}{\longrightarrow} \
			(\bl{u}_{k+t}, \bl{x}_{k+t})_{\T'} \
			\overset{\bl{x}}{\longrightarrow} \
			(\bl{x}_{k+t+1}, \bl{z}_{k+t+1})_{\T}.
	\]

	The case for $t$ being even can be similarly argued with due
	consideration for the fact that in this scenario the desired
	path will end in $\T' \neq \T$. \qed

\end{proof}

\begin{lemma}

	If an element $\bl{e}_k$ is $i$-bad w.r.t.\ table $\T$, then
	there exists a choice of the sets $\Su$ and $\X$, each of size
	at most $2i$, s.t.\ the block $\bl{e}$ cannot be stored in
	table $\T$.

	\label{app:lem:bad}

\end{lemma}

\begin{proof}

	As the element $\bl{e}_k$ is $i$-bad w.r.t.\ table $\T$, then,
	according to Definition~\ref{def:bad}, there exists some $j$
	between $1$ and $i$ s.t.\ two elements of the $j$-Universe of
	$\bl{e}_k$ w.r.t.\ table $\T$ belong to the same set in table
	$\T'$. Let the two elements of $\U^j_\T(\bl{e}_k)$ be
	$\bl{u}_{k+j}$ and $\bl{v}_{k+j}$. Then Lemma~\ref{lem:univ}
	suggests that there exist two paths, both of whom have
	$\bl{e}_k$ as the antecedent of the first node and are of
	length $j$, such that the antecedent of the last node of one
	of them is $\bl{u}_{k+j}$, as
	\[
		(\bl{e}_k, \bl{f}_k)_{\T} \
			\overset{\bl{f}}{\longrightarrow} \
			\dots \
			\overset{\bl{u}}{\longrightarrow} \
			(\bl{u}_{k+j}, \bl{x}_{k+j})_{\T'},
	\]
	and the antecedent of the last node of the other is
	$\bl{v}_{k+j}$, as
	\[
		(\bl{e}_k, \bl{g}_k)_{\T} \
			\overset{\bl{g}}{\longrightarrow} \
			\dots \
			\overset{\bl{v}}{\longrightarrow} \
			(\bl{v}_{k+j}, \bl{y}_{k+j})_{\T'}.
	\]
	The notion of the table $\T'$ in Definition~\ref{def:bad} and
	Lemma~\ref{lem:univ} are identical.

	It is given to us that the elements $\bl{u}_{k+j}$ and
	$\bl{v}_{k+j}$ belong to the same set in table $\T'$, and thus
	the last nodes of the two paths belong to the same set. In
	that case, we can set the consequents of the last two nodes of the paths as
	\[
		\bl{x}_{k+j} = \bl{v}_{k+j} \ \ \textnormal{ and }
			\ \ \bl{y}_{k+j} = \bl{u}_{k+j}.
	\]

	Suppose we store the element $\bl{e}_k$ in its own table,
	namely $\T$. According to Lemma~\ref{lem:path}, for the first
	path there exist a choice of sets $\Su_1$ and $\X_1$ such that
	the antecedent $\bl{u}_{k+j}$ must be stored in table $\T'$
	and the consequent $\bl{v}_{k+j}$ cannot be stored in $\T'$.
	Furthermore, as the length of the path is $j$, each of the
	subsets are of size $j$. For the second path we have sets
	$\Su_2$ and $\X_2$, of size $j$ each, such that the antecedent
	$\bl{v}_{k+j}$ will have to stored in table $\T'$.

	Let us consider the sets
	\[
		\Su = \Su_1 \cup \Su_2 \ \ \textnormal{ and }
			\ \ \X = \X_1 \cup \X_2,
	\]
	each of their sizes being at most $2j$, as some elements might be common among the sets. Together, it will imply that
	storing $\bl{e}_k$ in table $\T$ with result in a
	contradiction w.r.t.\ the storage of the element
	$\bl{v}_{k+j}$. So, we may conclude that the element
	$\bl{e}_k$, and consequently its block $\bl{e}$, cannot be
	stored in table $\T$.

	It may happen that an element in $\Su_1$ also is in $\X_2$,
	there by violating the requirement that the sets $\Su$ and
	$\X$ must be disjoint. Let that common element be
	$\bl{w}_{k+l}$, and the nodes of the two paths be
	\[
		(\bl{w}_{k+l}, \bl{x}_{k+l})_{\T''} \ \
			\textnormal{ and } \ \
			(\bl{w}_{k+l}, \bl{y}_{k+l})_{\T''},
	\]
	where the table $\T''$ will depend on $l$. Furthermore, we
	have $\bl{w}_{k+l}$ in both $\Su_1$ and $\X_2$, $\bl{x}_{k+l}$
	in $\X_1$ and $\bl{y}_{k+l}$ in $\Su_2$. In such a scenario,
	for the second node if we put $\bl{w}_{k+l}$ in $\Su_2$ and
	$\bl{y}_{k+l}$ in $\X_2$, the resulting sets $\Su$ and $\X$
	will become disjoint. As the index $k+l$ cannot occur anywhere
	else in the paths, this change will not affect the choices of
	other nodes in the two paths.

	Any other conflicts similar to the one above can be similarly
	resolved, thereby establishing the statement of the lemma. \qed

\end{proof}

In the next lemma, we will continue to use the notations $\bl{h}_{i,j}$ and $P_\T(\bl{e}_k)$ as introduced in the proof of Lemma~\ref{lem:univ_size}.

\begin{lemma}
	The following identity holds for sums indexed by the elements of
	$\U^t_\B(\bl{e}_1)$, where $\bl{e}_1$ is assumed to be
	$t$-good --
	\begin{multline*}
		\sum_{\bl{h}_{t,t+1} \ \in \ \U^t_\B(\bl{e}_1)} x_{\bl{h}_{t,t+1}}
			\ \ = \ \ \sum_{\bl{h}_{1,1} \ \in \ P_\B(\bl{e}_1)}
			\ \ \sum_{\bl{h}_{2,2} \ \in \ P_\C(\bl{h}_{1,2})}
			\ \ \sum_{\bl{h}_{3,3} \ \in \ P_\B(\bl{h}_{2,3})} \\
			\dots
			\ \ \sum_{\bl{h}_{t-1,t-1} \ \in \ P_\T(\bl{h}_{t-2,t-1})}
			\ \ \sum_{\bl{h}_{t,t} \ \in \ P_{\T'}(\bl{h}_{t-1,t})}
			x_{\bl{h}_{t,t+1}}.
	\end{multline*}
	Here, the tables $\T$ and $\T'$ are defined as follows.
	\[ \begin{array}{cl}
		\T = \C \textnormal{ and } \T' = \B, & \textnormal{if
			$t$ is odd}; \\
		\T = \B \textnormal{ and } \T' = \C, &
			\textnormal{otherwise.}
	\end{array} \]

	\label{app:lem:identity}
\end{lemma}

\begin{proof}

	We will prove the identity by induction on $t$. For the base
	case, when $t = 1$, we will establish the following --
	\[
		\sum_{\bl{h}_{1,2} \ \in \ \U^1_\B(\bl{e}_1)}
			x_{\bl{h}_{1,2}}
			\ \ = \ \ \sum_{\bl{h}_{1,1} \ \in \
				P_\B(\bl{e}_1)} x_{\bl{h}_{1,2}}.
	\]

	According to Lemma~\ref{lem:index}, the index of the elements
	in $\U^1_\B(\bl{e}_1)$ will be 2. If an element $\bl{h}_{1,2}$
	belongs to that universe, then from the definition of the
	universes of elements (Definition~\ref{def:univ}) we can
	deduce that $\bl{h}_{1,1} \in P_\B(\bl{e}_1)$. The identity
	consequently follows.

	Let us now assume that the identity holds for $t = k$, where
	$k$ is odd --
	\begin{multline*}
		\sum_{\bl{h}_{k,k+1} \ \in \ \U^k_\B(\bl{e}_1)}
			x_{\bl{h}_{k,k+1}}
			\ \ = \ \ \sum_{\bl{h}_{1,1} \ \in \ P_\B(\bl{e}_1)}
			\ \ \sum_{\bl{h}_{2,2} \ \in \ P_\C(\bl{h}_{1,2})}
			\ \ \sum_{\bl{h}_{3,3} \ \in \ P_\B(\bl{h}_{2,3})} \\
			\dots
			\ \ \sum_{\bl{h}_{k-1,k-1} \ \in \ P_\C(\bl{h}_{k-2,k-1})}
			\ \ \sum_{\bl{h}_{k,k} \ \in \ P_{\B}(\bl{h}_{k-1,k})}
			x_{\bl{h}_{k,k+1}}.
	\end{multline*}
	In the expression above $\bl{e}_1$ is $k$-good. We then see
	what happens in the case of $(k+1)$-universe with $\bl{e}_1$
	being $(k+1)$-good.
	\[
		\sum_{\bl{h}_{k+1,k+2} \ \in \ \U^{k+1}_\B(\bl{e}_1)}
			x_{\bl{h}_{k+1,k+2}}
			\ \ = \ \ \sum_{\bl{h}_{k+1,k+2} \ \in \
				\bigcup\limits_{\bl{h}_{k,k+1} \ \in \
				\U^k_\B(\bl{e}_1)} \U^1_\C(\bl{h}_{k,k+1})}
				x_{\bl{h}_{k+1,k+2}}
	\]
	The R.H.S.\ comes from Definition~\ref{def:univ}. As the
	element $\bl{e}_1$ is $(k+1)$-good, and consequently $k$-good,
	the elements of $\U^k_\B(\bl{e}_1)$ all belong to distinct
	sets in table $\C$, and hence all the 1-universes computed
	w.r.t.\ table $\C$ are disjoint from each other. Thus we can
	rewrite the R.H.S.\ as follows.
	\begin{align*}
		\sum_{\bl{h}_{k+1,k+2} \ \in \ \U^{k+1}_\B(\bl{e}_1)}
			x_{\bl{h}_{k+1,k+2}}
			\ \ &= \ \ \sum_{\bl{h}_{k,k+1} \ \in \
				\U^k_\B(\bl{e}_1)} \left(
				\sum_{\bl{h}_{k+1,k+2} \ \in \
				\U^1_\C(\bl{h}_{k,k+1})} x_{\bl{h}_{k+1,k+2}}
				\right) \\
			&= \ \ \sum_{\bl{h}_{k,k+1} \ \in \
				\U^k_\B(\bl{e}_1)} \left(
				\sum_{\bl{h}_{k+1,k+1} \ \in \
				P_\C(\bl{h}_{k,k+1})} x_{\bl{h}_{k+1,k+2}}
				\right)
	\end{align*}
	The R.H.S.\ is a sum indexed by the elements of
	$\U^k_\B(\bl{e}_1)$, and by mathematical induction we can
	rewrite it as --
	\begin{multline*}
		\sum_{\bl{h}_{k+1,k+2} \ \in \ \U^{k+1}_\B(\bl{e}_1)}
			x_{\bl{h}_{k+1,k+2}}
			\ \ = \ \ \sum_{\bl{h}_{1,1} \ \in \ P_\B(\bl{e}_1)}
			\ \ \sum_{\bl{h}_{2,2} \ \in \ P_\C(\bl{h}_{1,2})}
			\ \ \sum_{\bl{h}_{3,3} \ \in \ P_\B(\bl{h}_{2,3})} \\
			\ \ \dots
			\ \ \sum_{\bl{h}_{k-1,k-1} \ \in \ P_\C(\bl{h}_{k-2,k-1})}
			\ \ \sum_{\bl{h}_{k,k} \ \in \ P_{\B}(\bl{h}_{k-1,k})}
			\left( \sum_{\bl{h}_{k+1,k+1} \ \in \
				P_\C(\bl{h}_{k,k+1})} x_{\bl{h}_{k+1,k+2}}
				\right).
	\end{multline*}

	The case when $k$ is even can be similarly argued. \qed

\end{proof}

\end{document}